\documentclass[11pt]{article}
\usepackage{amsmath,amssymb,amsfonts,amsthm,epsfig}
\usepackage[usenames,dvipsnames]{xcolor}
\usepackage{bm,xspace}
\usepackage{tcolorbox}
\usepackage{cancel}
\usepackage{fullpage}
\usepackage{liyang}
\usepackage{framed}
\usepackage{verbatim}
\usepackage{enumitem}
\usepackage{array}
\usepackage{multirow}
\usepackage{afterpage}
\usepackage{mathrsfs}
\usepackage{pifont} 
\usepackage{chngpage}
\usepackage[normalem]{ulem}
\usepackage{boxedminipage}
\usepackage{caption}
\usepackage{forest}

\usepackage{algorithm}
\usepackage{algorithmicx}
\usepackage{algpseudocode}

\usepackage{pgfplots}
\pgfplotsset{width=8cm,compat=newest}

\usepackage{tikz}
\usetikzlibrary{calc,through,backgrounds}
\usepackage{tikz-cd}

\def\colorful{0}

\ifnum\colorful=1
\newcommand{\violet}[1]{{\color{violet}{#1}}}

\newcommand{\blue}[1]{{{\color{blue}#1}}}
\newcommand{\red}[1]{{\color{red} {#1}}}

\fi
\ifnum\colorful=0
\newcommand{\violet}[1]{{{#1}}}

\newcommand{\blue}[1]{{{#1}}}
\newcommand{\red}[1]{{{#1}}}

\fi

\newcommand{\mcA}{\mathcal{A}}

\newcommand{\error}{\mathrm{error}}

\newcommand{\updownarrows}{\uparrow\mathrel{\mspace{-1mu}}\downarrow}

\makeatletter
\newtheorem*{rep@theorem}{\rep@title}
\newcommand{\newreptheorem}[2]{
\newenvironment{rep#1}[1]{
 \def\rep@title{#2 \ref{##1}}
 \begin{rep@theorem}\itshape}
 {\end{rep@theorem}}}
\makeatother

\newreptheorem{theorem}{Theorem}

\begin{document}

\title{A query-optimal algorithm for finding counterfactuals
\vspace{15pt}}

\author{Guy Blanc \vspace{8pt} \\ \hspace{-5pt}{\sl Stanford}
\and Caleb Koch \vspace{8pt}\\ \hspace{-8pt} {\sl Stanford}
\and \hspace{0pt} Jane Lange \vspace{8pt} \\ \hspace{-4pt}  {\sl MIT}
\and Li-Yang Tan \vspace{8pt} \\ \hspace{-8pt} {\sl Stanford}}

\date{\vspace{15pt}\small{\today}}

\maketitle

\begin{abstract}
We design an algorithm for finding counterfactuals with strong theoretical guarantees on its performance.  For any monotone  model $f : X^d \to \zo$ and instance $x^\star$, our algorithm makes
\[ \red{S(f)^{O(\Delta_f(x^\star))}\cdot \log d}\]
queries to $f$ and returns \violet{an {\sl optimal}} counterfactual for $x^\star$: a nearest instance $x'$ to $x^\star$ for which $f(x')\ne f(x^\star)$.  Here $S(f)$ is the sensitivity of $f$, a discrete analogue of the Lipschitz constant, and $\Delta_f(x^\star)$ is the distance from $x^\star$ to its nearest counterfactuals. The previous best known query complexity  was $d^{\,O(\Delta_f(x^\star))}$, achievable by brute-force local search. 
We further prove a lower bound of $S(f)^{\Omega(\Delta_f(x^\star))} + \Omega(\log d)$ on the query complexity of any algorithm, thereby showing that the  guarantees of our algorithm are essentially optimal. 
\end{abstract} 

\thispagestyle{empty}
\newpage

\section{Introduction}

Counterfactual reasoning is at the very heart of  causal inference~\cite{Per09,Per09book,MW15}.  Counterfactuals are answers to ``what would have happened if" questions: if a person has been denied a loan, would their application have been approved if their annual income were \$$20$k higher?   In explainable ML, there is a growing interest in the use of {\sl counterfactual explanations}~\cite{WMR17} to understand the predictions of black box models: given a model $f$ and an instance $x^\star$, we would like to efficiently determine how a few features of $x^\star$ can be changed to obtain a similar instance~$x'$ for which $f(x')\ne f(x^\star)$.

It is natural to seek {\sl sparse} counterfactuals, meaning that $x'$ differs from $x^\star$ in as few features as possible.  Ideally, we would like counterfactuals that are {\sl optimally} sparse:

\begin{definition}[Sparsity and optimality] 
For a model $f : X^d \to\zo$ and two instances $x^\star,x'$ such that $f(x^\star)\ne f(x')$, the {\sl sparsity} of $x'$ as a counterfactual for $x^\star$ is the number of features on which they differ: \violet{the quantity $|\Delta(x^\star,x')|$, where} 
\[ \Delta(x^\star,x') \coloneqq \{ i\in [d]\colon x^\star_i \ne x'_i\}. \] 
The {\sl counterfactual complexity} of $x^\star$ with respect to $f$ is the quantity: 
\[ \Delta_f(x^\star) \coloneqq \min_{x'\in X^d} \{ \violet{|\Delta(x^\star,x')|} \colon f(x^\star)\ne f(x')\}, \] and we say that $x'$ is an {\sl optimal} counterfactual for $x^\star$ if $\violet{|\Delta(x^\star,x')|}= \Delta_f(x^\star)$. 
\end{definition} 

Another desideratum that has received significant attention, motivated by the need for actionable recourse~\cite{USL19}, is that of {\sl diversity} in counterfactual explanations~\cite{WMR17,Rus19,MST20,KBBV20}: a wide range of counterfactuals instead of just a single one.  

\subsection{Our contributions}

\violet{Our first result is an efficient algorithm for finding {\sl all optimal} counterfactuals for any monotone $f$ in the setting of binary features:

\begin{theorem}
\label{thm:main-intro} 
Given queries to a monotone model $f : \zo^d \to \zo$ with sensitivity $S(f)$ and an instance $x^\star$, our algorithm makes 
\[ \red{S(f)^{O(\Delta_f(x^\star))}\cdot \log d}\]
queries to $f$ and w.h.p.\footnote{Throughout we write ``with high probability'' or w.h.p.~to indicate probability at least $1-1/\text{poly}(d).$}~returns all optimal counterfactuals for $x^\star$. 
\end{theorem}

We will define and discuss sensitivity in the body of the paper (see~\Cref{sec:prelims}), mentioning for now that it is a well-studied discrete analogue of the Lipschitz constant, \violet{and this quantity is often much smaller than $d$.}   

Next, we consider the setting of general feature spaces.  In this setting it is infeasible to return {\sl all} optimal counterfactuals due to the sheer number of them.  We are nevertheless able to efficiently return the collection of all {\sl subsets of features} induced by these optimal counterfactuals:    

\begin{theorem} 
\label{thm:main-intro-non-boolean} 
Given queries to a monotone model $f : X^d \to \zo$ with sensitivity $S(f)$ and an instance $x^\star$, our algorithm makes 
\[ \red{S(f)^{O(\Delta_f(x^\star))}\cdot \log d}\]
queries to $f$ and w.h.p.~returns the collection \violet{$\mathcal{C}_f(x^\star) = \{ \Delta(x^\star,x') \colon \text{$x'$ is an optimal counterfactual for $x^\star$}\}$. }
\end{theorem}

} 




\Cref{thm:main-intro,thm:main-intro-non-boolean} give the first algorithms with query complexity that \violet{evades the curse of dimensionality, and indeed, strongly so}.  
We contrast our query complexity to that of ball search, a simple and natural algorithm for finding counterfactuals: first query $f$ on $x^\star$ and all instances that differ from $x^\star$ by a single feature. (By the monotonicity of $f$, for any feature, it suffices to query $f$ on the instance that differs maximally from $x^\star$ on that feature.)  Next, query $f$ on the instances that differ by two features, and so on, until a counterfactual is found.  This algorithm has query complexity $d^{\,O(\Delta_f(x^\star))}$, an exponentially worse dependence on $d$  than ours.  Prior to our work, this was the best known query complexity even just to return a {\sl single} optimal counterfactual.  

\paragraph{Lower bounds.} We complement~\Cref{thm:main-intro,thm:main-intro-non-boolean} with a couple of lower bounds.  All our lower bounds hold even in the setting of binary features $(X= \zo)$, and \violet{against randomized algorithms that are only required to successfully return an optimal counterfactual with a small probability (say $0.01$)}.  We first show that the query complexity of our algorithm is essentially optimal: 

\violet{\begin{theorem}[Optimality of~\Cref{thm:main-intro,thm:main-intro-non-boolean}, see \Cref{thm:lb-monotone} for the formal version]
\label{thm:lb-intro}
For any algorithm $\mcA$ and $S \in \N$, there is a monotone model $f : \zo^d\to \zo$ with $S(f) = S$ and an instance $x^\star$ such that $\mcA$ must make  
\[ S^{\,\Omega(\Delta_f(x^\star))} + \violet{\Omega}(\log d)\] 
many queries to $f$, even just to find a single optimal counterfactual for $x^\star$.
\end{theorem} }




We also establish the inherent limitations of {\sl local search} algorithms, a broad and natural class of algorithms for finding counterfactuals.  A local search algorithm is any algorithm whose first query is $x^\star$, and whose every subsequent query differs from a previous one by exactly one feature.  For example, ball search is a local search algorithm.  We show that no local search algorithm can achieve the query complexity that our algorithm does, even for low-sensitivity monotone functions:

\violet{
\begin{theorem}[Lower bound for local search algorithms, see \Cref{thm:lb-local-formal} for the formal version]\label{thm:local-lb} 
    For any local search algorithm $\mcA$ and $d \in \N$, there is a monotone model $f : \zo^d\to \zo$ with $S(f) = 1$ and an instance $x^\star$ with $\Delta_f(x^\star) = 1$, such that $\mcA$ must make $\Omega(d)$ many queries to $f$ to find a single optimal counterfactual for $x^\star$.
\end{theorem}}

\subsection{Overview of our algorithm and techniques}

\Cref{thm:main-intro,thm:main-intro-non-boolean} circumvent the lower bound of~\Cref{thm:local-lb} via a novel approach that is {\sl not} based on local search.  The crux of our approach is a new algorithm for understanding monotone black box models that we believe will see further utility beyond counterfactual explanations: using only queries to a monotone black box model~$f$, this algorithm allows us to navigate a decision tree representation $T$ of $f$ without actually building $T$ in full.  Crucially, since we strive to handle arbitrarily complex models $f$, this tree $T$ may be too large to build efficiently in its entirety. 

With this algorithm, given any instance $x^\star$, we will build only the ``necessary part" of $T$ to find all optimal counterfactuals for $x^\star$.  We show that just a tiny portion of $T$ suffices for this purpose: it suffices to build only the root-to-leaf path $\rho$ in $T$ that $x^\star$ follows and the paths that are ``close to" $\rho$ in $T$ \violet{in a sense that we will make precise (see the beginning of~\Cref{sec:proof-of-main} for more details on this notion of `` path distance")}. Writing $h$ to denote the depth of $T$, we prove that we only have to construct \red{$h^{O(\Delta_f(x^\star))}$} many nodes in $T$, which is only an exponentially small part of $T$.  Next, by bringing together classical~\cite{Nis89} and recently-develop results~\cite{BKLT22} about the structure of monotone functions, we are able to bound the depth of $T$ and in terms of $f$'s sensitivity.

To summarize, our overall approach allows us to enjoy the benefits of having access to a highly interpretable representation of a black box model $f$---a decision tree representation---without paying the price associated with $T$ being intractably large for most complex models of interest.

\subsection{Related work} 

Our work is theoretical in nature: we give an algorithm with strong bounds on its query complexity and all the counterfactuals that it returns are guaranteed to be optimal.  We view our main contribution as proving \violet{the curse of dimensionality can be strongly evaded} for a broad and natural class of models---all monotone models. 

There is a substantial body of empirical work on counterfactual explanations.  The algorithms in these works either do not guarantee optimal counterfactuals or do not come with a formal analysis of their query complexity.  Many of these works rely on classic optimization tools:  the algorithms of~\cite{WMR17,MST20} are based on gradient descent and hence are restricted to differentiable models, whereas the algorithms of~\cite{Rus19,USL19} use mixed integer programming and are limited to linear models.  \cite{KBBV20} encode the problem of finding optimal counterfactuals as a satisfiability problem, which they then solve using  SMT solvers.

The importance of diversity in counterfactual explanations has been highlighted in several works~\cite{WMR17,Rus19,MST20,KBBV20}. Significant motivation comes from the need for actionable recourse~\cite{USL19}: a large set of counterfactuals is more likely to contain one that is actionable (e.g.~in the context of a loan denial, an applicant can plausibly work towards earning a higher income but cannot change their history of defaults).  

Regarding our techniques, they are in the spirit of a recent line of \violet{theoretical} work on {\sl implicit} learning algorithms: algorithms that are able to access their hypotheses efficiently without constructing them in full~\cite{KV18,BH18,KVB20,BBG20,BLT21neurips}.  In particular, Blanc, Lange, and Tan~\cite{BLT21neurips} show how, given queries to a model $f$, one can efficiently navigate an implicit decision tree hypothesis $T$ that is $\eps$-close to $f$ under the uniform distribution.  The presence of errors in the hypothesis, and the fact that these errors are measured with respect to the uniform distribution, are significant limitations of their result.  Our work shows how these limitations can be overcome, in a strong sense, in the case of monotone models $f$: our implicit decision tree hypothesis $T$ computes $f$ {\sl exactly}, which is crucial for our application to finding counterfactuals.    Another limitation of~\cite{BLT21neurips}'s analysis is that it only applies to the setting of binary features; we do not need this assumption. 

Finally, we mention that counterfactual explanations are part of a broader landscape of local  explanations~\cite{SK10,BSHKHM10,SVZ14,RSG16,KL17,LL17,RSG18}, which seek to explain a model's prediction for specific inputs.  Global explanations, on the other hand, approximate the entire model with a simple and interpretable one~\cite{CS95,BS96,AB07,ZH16,VLJODV17,BKB17,VS20,LKCL19,LAB20}.

\subsection{Preliminaries} 
\label{sec:prelims} 
We rely on a few standard notions from the study of Boolean functions.

\violet{ 
\begin{definition}[Sensitivity]
For a function $f : X^d \to \zo$ and an instance $x\in X^d$, the {\sl sensitivity of $f$ at $x$} is the quantity 
\[ S_f(x) = |\{ i \in [d] \colon\, \exists\, a \in X\ \text{s.t.}\ f(x) \ne f(x_{i\leftarrow a})\}|, \] 
where $i\leftarrow a$ denotes $x$ with its $i$-th feature set to $a$.  That is, $S_f(x)$ is the number of features $i$ of $x$ for which there is a way of changing $x_i$ that flips $f$'s value on $x$.  

The {\sl sensitivity of $f$} is the quantity $\ds S(f) = \max_{x\in X^d}\{ S_f(x)\}.$ 
\end{definition}

The sensitivity of a function can be viewed as a discrete analogue of the Lipschitz constant, with low-sensitivity discrete functions being considered smooth.  To see this analogy, we observe that for all $x \in X^d$ and $\delta \in (0,1)$, 
\[ \mathop{\Ex_{\by\in X^d}}_{\Delta(x,\by)=\delta d}\big[|f(x)-f(\by)|\big] \le \delta\cdot S(f). \] 
See~\cite{GNSTW16} for more on this perspective.  Introduced by Cook, Dwork, and Reischuk~\cite{CDR86}, the sensitivity of discrete functions is the subject of intensive study in theoretical computer science; it is, for example, the key notion in the recent breakthrough {\sl Sensitivity Theorem} of Huang~\cite{Hua19} (see also~\cite{Knu19}), resolving a longstanding conjecture of Nisan and Szegedy~\cite{NS94}. 
}

\paragraph{Monotonicity.}{
For a model $f:X^d\to\zo$, we assume an ordering $\le$ on the elements of $X$ which we lift coordinatewise to a partial ordering on $X^d$. That is, for $x,y\in X^d$, $x\le y$ if and only if $x_i\le y_i$ for all $i\in[d]$. We say $f$ is \textit{monotone} if it is monotone with respect to such an ordering. We'll also assume that $X$ has a maximum and minimum element.
}

\paragraph{Structure of the rest of this paper.} We first prove \Cref{thm:main-intro}, which focuses on functions with Boolean features. To do so, we will recall the notion of an implicit decision tree from \cite{BLT21neurips} in \Cref{sec:implicit}. Then, in \Cref{sec:proof-of-main}, we prove \Cref{thm:main-intro} assuming a sufficiently good implicit decision tree exists. In \Cref{sec:build-idt}, we complete the proof of \Cref{thm:main-intro} by showing how to build that implicit decision tree. 

We prove our lower bounds, \Cref{thm:lb-intro,thm:local-lb}, in \Cref{sec:lb}. In \Cref{sec:general-feature}, we show how to extend our algorithm to the setting of arbitrary features by reducing to the Boolean setting and prove \Cref{thm:main-intro-non-boolean}.

\section{Implicit decision trees for monotone models} 
\label{sec:implicit}
Our algorithm for finding counterfactuals will rely on having access to a decision tree $T$ which exactly represents the model $f:\zo^d\to\zo$.  At a high level, for an instance $x^\star$, the algorithm follows a unique root-to-leaf path $\rho$ in $T$ searches all root-to-leaf paths in $T$ that are ``close to"  $\rho$ in a sense that make precise in the next section.  In general, the size of $T$ will be exponentially larger than the search space of our counterfactual finding algorithm, and hence it is advantageous to avoid computing the entire tree. Instead, we maintain an \textit{implicit} copy of the tree $T$. This copy allows us to compute only the parts of the tree that we need to find optimal counterfactuals for $x^\star$.

\paragraph{Restrictions.}{An implicit decision tree $T$, a notion introduced in~\cite{BLT21neurips}, is an algorithm which given query access to $f$ can navigate $T$ without building it in full. A node in $T$ is specified by a \textit{restriction} which is a partial function $\rho:[d]\rightharpoonup\{0,1\}$ indicating which (if any) features are fixed to specific values. These features and their values correspond to those queried along a path in the tree. We write $\text{Dom}(\rho)\sse [d]$ to denote the domain of the restriction. The size of a restriction $|\rho|$ is $|\text{Dom}(\rho)|$, the number of values on which it is defined.} We write $f_\rho:\zo^{d-|\text{Dom}(\rho)|}\to\zo$ to denote the restriction of $f$ to $\rho$. That is, $f_\rho(x)$ is $f$ evaluated on the instance $x'\in\zo^d$ which is formed by inserting features specified by $\rho$ into the instance $x$. For example, $\rho=\{1\mapsto 1,4\mapsto 1\}$ is the restriction where the $1$st and $4$th features are fixed to $1$. If $f:\zo^4\to\zo$, then $f_\rho:\zo^2\to\zo$ and e.g. $f_\rho(00)=f(1001)$ under this restriction.

\begin{definition}[Implicit decision tree]
\label{def:implicit}
    An implicit decision tree (IDT) \blue{$T$} for $f:\zo^d\to\zo$ is an algorithm which has query access to $f$ and supports the following operations. For a restriction $\rho:[d]\rightharpoonup \zo$ corresponding to features queried along a root-to-$\rho$ path in $T$ and auxiliary information $A\sse [d]$:
    \begin{enumerate}
        \item {\sc IsLeaf}$_f(\rho, A)$ indicates whether $\rho$ is a leaf in $T$.
        \item {\sc Query}$_f(\rho, A)$, for a non-leaf $\rho$, returns $i\in [d]$ corresponding to the index queried at $\rho$ in $T$ and $A'\subseteq [d]$, updated auxiliary information.
    \end{enumerate} 
    If each operation uses at most $q$ queries to $f$ with high probability then the algorithm is a $q$-query implicit decision tree for $f$.\footnote{Note that the term ``query'' here is being used in two separate ways. A variable being ``queried'' on a path simply means that variable appears as an internal node on that path. The IDT itself has ``query access'' to $f$ meaning it may query the value of $f(x)$ for various $x$ in its implementation.}
\end{definition}


\paragraph{The role of auxiliary information.}{ Our algorithm for finding counterfactuals updates the auxiliary information after each query and we assume that the auxiliary information being passed to the algorithm originates from the previous query (the auxiliary information at the root can be arbitrary). Internally, the auxiliary information will be a subset of features and will allow us to precompute parts of the tree instead of node by node. This way, most calls to {\sc Query}$(\rho, A)$ will simply return the appropriate feature index from $A$ until all such features in $A$ are exhausted after which a new $A$ is computed from scratch. The example below illustrates the reuse of auxiliary information between IDT operations. For a more detailed illustration of an IDT and the precomputation of parts of the tree using auxiliary information see \Cref{fig:path_thru_idt}.}

\paragraph{Example.}{
\Cref{fig:IDT_query_example} illustrates how the IDT operations can be used to walk down a decision tree $T$ representing a function $f$. In this case, we make three calls to {\sc Query}$_f(\cdot,\cdot)$ corresponding to the depth of this particular root-to-leaf path. In general, if a $q$-query IDT has depth at most $k$, then the root-to-leaf path corresponding to an instance $x^\star$ can be constructed using $O(k)$ calls to the IDT operations and therefore $O(kq)$ queries to $f$. 

\begin{figure}[h]
\begin{center}
\forestset{
 default preamble={
 for tree={
  circle,
  l sep=0.8cm,
  s sep=0.5cm,
  scale=1.0
 }
 }
}
\scalebox{1}{
\begin{forest}
 [$\color{blue}x_1$,name=x1,
    [$x_2$,name=x2
        [$x_4$,name=x4l
            [$0$]
            [$1$]
        ]
        [$x_5$,name=x5l
            [$0$]
            [$1$]
        ]
    ]
    [$\color{blue}x_3$,name=x3
        [$\color{blue}x_5$,name=x5r
            [$\color{blue}0$,name=x5r0]
            [$1$]
        ]
        [$x_4$,name=x4r
            [$0$]
            [$1$]
        ]
    ]
 ]
 \draw[line width=0.35mm, blue] (x1) to (x3);
 \draw[line width=0.35mm, blue] (x3) to (x5r);
 \draw[line width=0.35mm, blue] (x5r) to (x5r0);
\draw[color=black] (.east)+(10em,-4em) node[text width=7cm] (label) [right] {
    \footnotesize 
    \begin{enumerate}
        \item {\sc IsLeaf}$_f(\{\}, \{\})=\text{No}$
        \item {\sc Query}$_f(\{\}, \{\})=(1,A_1)$
        \item {\sc IsLeaf}$_f(\{1\mapsto 1\}, A_1)=\text{No}$
        \item {\sc Query}$_f(\{1\mapsto 1\}, A_1)=(3,A_2)$
        \item {\sc IsLeaf}$_f(\{1\mapsto 1,3\mapsto 0\}, A_2)=\text{No}$
        \item {\sc Query}$_f(\{1\mapsto 1,3\mapsto 0\}, A_2)=(5,A_3)$
        \item {\sc IsLeaf}$_f(\{1\mapsto 1,3\mapsto 0,5\mapsto 0\}, A_3)=\text{Yes}$
    \end{enumerate}
    };
\end{forest}
}
\end{center}
\caption{Walking down a tree $T$ representing a function $f:\{0,1\}^5\to\{0,1\}$ using the IDT operations. The input is $x^\star=10000$ and its computation path through the tree is colored in blue. The walk starts at the root with the empty restriction $\{\}$ and ends at the leaf node represented by the restriction $\{1\mapsto 1,3\mapsto 0,5\mapsto 0\}$. The decision tree leaf value here is $0$ indicating that $T(x^\star)=f(x^\star)=0$.}
\label{fig:IDT_query_example}
\end{figure}
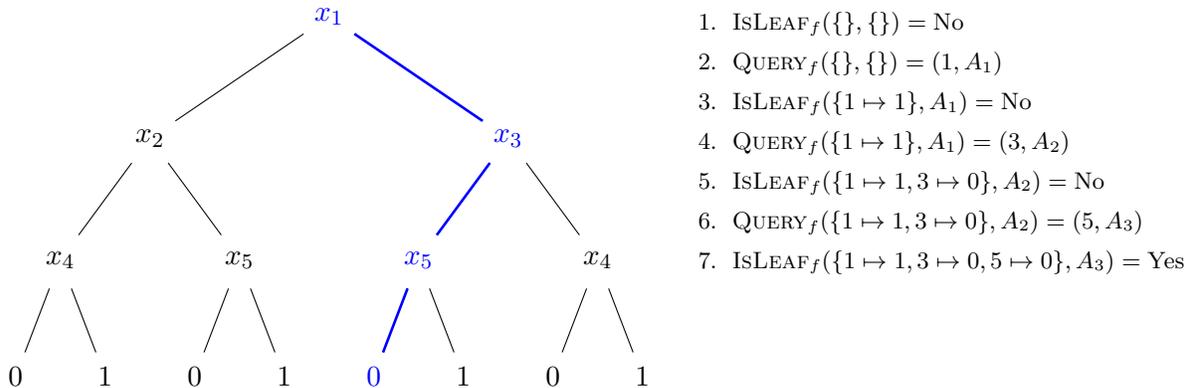
}

\paragraph{Our key technical lemma.} For our intended application, it will be important that the IDT {\sl exactly} represents $f$.  The work of~\cite{BLT21neurips} design a relaxed variant of IDTs that only {\sl approximates} models $f$ with respect to the uniform distribution over instances.  Both the presence of errors and the uniform-distribution assumption are not ideal, and can be seen to be inherent to the proof technique of~\cite{BLT21neurips}.  In this work we use recently developed structural results for monotone models~\cite{BKLT22} to overcome these limitations.  Our key technical lemma is the following.  

\begin{lemma}[Exact IDTs for monotone models]
    \label{lem:efficient_IDTs}
    Let $f:\zo^d\to\zo$ be monotone. Then there is an $S(f)^{O(1)}\cdot\log d$-query IDT for $f$ with depth at most $S(f)^{O(1)}$. 
\end{lemma}

In the next section we prove~\Cref{thm:main-intro} assuming~\Cref{lem:efficient_IDTs}.  We then prove~\Cref{lem:efficient_IDTs} in~\Cref{sec:proof-of-lemma}.


\section{Using IDTs to find optimal counterfactuals: Proof of~\texorpdfstring{\Cref{thm:main-intro}}{Theorem 1}}

\label{sec:proof-of-main}

Suppose we may access an implicit decision tree $T$ for $f$ as described in \Cref{def:implicit}. We present an algorithm for finding counterfactuals using these operations and analyze its query complexity in terms of the depth $h$ of $T$ \violet{and the query complexity of the IDT}---\Cref{lem:efficient_IDTs} provides strong bounds on both quantities in terms of the sensitivity of $f$.  


\violet{\paragraph{Intuition.} Given an instance $x^\star$, it follows a unique root-to-leaf path $\rho$ in $T$.  We begin with a simple observation: if there is a counterfactual $x'$ of distance $1$ from $x^\star$, it must belong to the set $I_1$ of all instances that differ from $x^\star$ on a single feature {\sl queried along $\rho$}.  The size of $I_1$ is the length of $\rho$, which is at most the depth $h$ of $T$.  Building on this observation, if $x''$ is a counterfactual of distance $2$ from $x^\star$, there must be an instance $x'\in I_1$ such that $x''$ differs from $x'$ on a single feature queried on the path $\pi$ that $x'$ follows in $T$.  This reasoning leads to a bound of $h^{O(\Delta_f(x^\star))}$ on the size of our search space---all paths in $T$ of ``path distance" $\Delta_f(x^\star)$ from $\rho$---which stands in contrast to $d^{O(\Delta_f(x^\star))}$, the size of the search space of ball search.}





\violet{\subsection{A helper algorithm that finds minimal counterfactuals}  

Our algorithm for finding optimal countefactuals will call on a subroutine, {\sc FindMinimal} (\Cref{fig:FindMinimal}), for finding {\sl minimal} counterfactuals:} 

\begin{definition}[Minimal counterfactual]
For a model $f: \zo^d\to\zo$ and an instance $x^\star$, a {\sl minimal counterfactual for $x^\star$ among $I \sse \zo^d$} is an instance $x' \in I$ such that:
\begin{enumerate} 
\item $f(x') \ne f(x^\star)$, and 
\item Let $V\sse [d]$ denote the features that $x^\star$ and $x'$ differ on.  Then $f(x'') = f(x^\star)$ for all instances $x''\in I$ that differ from $x^\star$ only on the features in a strict subset $U\subset V$. 
\end{enumerate} 
\end{definition}

\violet{We observe that optimal counterfactuals are minimal, but minimal counterfactuals may not be optimal.}



\violet{Our pseudocode for {\sc FindMinimal} uses the following notation.}  Recall that a restriction is a partial function $\rho : [d] \rightharpoonup \zo$, which we can equivalently represent as a string $\rho \in \{0, 1, *\}^d$. We use the notation $\rho_{i \gets b}$ to denote the restriction $\rho$ with $\rho_i$ overwritten with $b$.  For $x^\star\in \zo^d$, we define $x^\star$ overwritten by $\rho$, denoted $x^\star_\rho \in \zo^d$, to be the following hybrid input: for all $i\in [d]$,
\[ (x^\star_\rho)_i = \begin{cases}
 x^\star_i & \text{if $\rho_i = *$} \\
\rho_i & \text{if $\rho_i \ne *$.}
\end{cases}\]

\begin{figure*}[h]
  \captionsetup{width=.9\linewidth}
\begin{tcolorbox}[colback = white,arc=1mm, boxrule=0.25mm]
\vspace{3pt}

$\textsc{FindMinimal}_f(\rho, x^\star, k, A)$:
\begin{itemize}[align=left]
	\item[\textbf{Given:}] Access to $f$ via  queries and IDT operations \violet{for a tree $T$ that represents $f$}.  Restriction $\rho$ \violet{corresponding to a path in $T$}, instance $x^\star$, distance bound $k$, auxiliary information $A$. 
	\item[\textbf{Output:}] All minimal counterfactuals $x'$ for $x^\star$ among those that are consistent with $\rho$ and satisfy $|\Delta_f(x',x^\star_\rho)| \le k$. 
\end{itemize}
\begin{enumerate}
	\item If $f(x^\star) \ne f(x^\star_\rho)$, return $\{x^\star_\rho\}$. 
    \item Else if  $\textsc{IsLeaf}(\rho, A)$ or $k = 0$, return $\varnothing$.
    \item Else: 
    \begin{enumerate}
        \item Let $(i, A') \gets \textsc{Query}(\rho, A)$.
		\item Return $\textsc{FindMinimal}_f(\rho_{i \gets x^\star_i}, x^\star, k, A') ~\cup~ \textsc{FindMinimal}_f(\rho_{i \gets \overline{x}^\star_i}, x^\star, k-1, A')$.
    \end{enumerate}
\end{enumerate}

\end{tcolorbox}
\caption{\violet{Helper algorithm $\textsc{FindMinimal}$}}
\label{fig:FindMinimal}
\end{figure*}

\subsubsection{Analysis of {\sc FindMinimal}}

\violet{
\begin{lemma} [Correctness of {\sc FindMinimal}]
\label{lem:correctness}
$\textsc{FindMinimal}_f(\rho, x^\star, k, A)$ returns all minimal counterfactuals $x'$ for $x^\star$ among those that are consistent with $\rho$ and satisfy $|\Delta_f(x',x^\star_\rho)|\le k$. 
\end{lemma} 

\begin{proof} 
We begin by noting all instances that are consistent with $\rho$ differ from $x^\star$ on at least the features on which $x^\star$ and $x^\star_\rho$ differ.  Therefore, if $f(x^\star_\rho)\ne f(x^\star)$, {\sc FindMinimal} correctly returns $\{ x^\star_\rho\}$ in Line 1. 

We therefore assume that $f(x^\star_\rho) = f(x^\star)$ and proceed by induction on the height of the subtree of $T$ rooted at $\rho$.  If $\rho$ is a leaf, then $f$ takes the same value on all instances $x'$ that are consistent with $\rho$, and hence $f(x') = f(x^\star_\rho) = f(x^\star)$.  In this case {\sc FindMinimal} correctly returns $\varnothing$ in Line 2. 

For the inductive step, suppose $\rho$ is an internal node.  If $k=0$, the only instance that is of distance $0$ from $x^\star_\rho$ is $x^\star_\rho$ itself, so {\sc FindMinimal} again correctly returns $\varnothing$ in Line 2.  Otherwise, the search space of all instances $x'$ that are consistent with $\rho$ and satisfy $\Delta_f(x',x^\star_\rho)\le k$ can be partitioned into: 
\begin{itemize} 
\item[$\circ$] Those that are consistent with $\rho_{i\leftarrow x_i^\star}$ and satisfy $|\Delta_f(x',x^\star_{\rho_{i\leftarrow x^\star_i}})| \le k$ 
\item[$\circ$] Those that are consistent with $\rho_{i\leftarrow \overline{x}_i^\star}$ and satisfy $|\Delta_f(x',x^\star_{\rho_{i\leftarrow \overline{x}^\star_i}})| \le k-1$,
\end{itemize} 
where $i$ is the feature that is queried at $\rho$ (Line 3a). {\sc FindMinimal} recurses on these two sets in Line 3b, and correctness follows by our induction hypothesis.
\end{proof} 
}

\begin{remark}[Correctness of auxiliary information] 
Our proof of~\Cref{lem:correctness} relies on the correctness of the IDT operations, which only hold if our calls to them use the  ``correct" auxiliary information~$A$.  This will be the case given the way we use {\sc FindMinimal} in our overall algorithm {\sc FindOptimal}. {\sc FindOptimal} calls  {\sc FindMinimal} with $\rho$ being the empty restriction and $A = \varnothing$.  In its recursive execution, {\sc FindMinimal} with auxiliary information $A$ recursively calls itself with $A'$ where $A'$ is the auxiliary information returned by $\textsc{Query}(\rho,A)$ in Line 3a: the correctness of $A'$ follows from the correctness of $A$ and the {\sc Query} operation. 
\end{remark}




\begin{lemma}[Query complexity of {\sc FindMinimal}]
\label{lem:runtime}
	$\textsc{FindMinimal}_f(\rho, x, k, A)$ makes $h^{O(k)}$ calls to the IDT operations and $h^{O(k)}$ additional queries to $f$, where $h$ is the height \violet{of the subtree of $T$ rooted at $\rho$}.  
\end{lemma}

\begin{proof}
We claim that both quantities are  bounded by $Q(h,k)\le 2(h+1)^k$.
The proof is by induction on $h$. When $h =0$ (i.e.~$\rho$ is a leaf in $T$), {\sc FindMinimal} queries $f$ on $x^\star$ and $x^\star_\rho$, makes one call to {\sc IsLeaf}, and terminates.  When $h\ge 1$, there are at most two queries to $f$ (on $x^\star$ and $x^\star_\rho$), two calls to IDT operations ({\sc IsLeaf} and {\sc Query}), in addition to those made recursively.   We therefore have the following recursive relation: 
\[ Q(h,k) \le 2 + Q(h-1,k) + Q(h-1,k-1). \] 
Applying the induction hypothesis, we can bound $Q(h,k)$ by 
\begin{equation*} 
    Q(h,k) \le 2 + 2h^k + 2h^{k-1} \le 2(h+1)^k. \qedhere
\end{equation*}






\end{proof}

\subsection{Our algorithm for finding optimal counterfactuals} 
\label{sec:find-optimal} 

With the helper algorithm {\sc FindMinimal} in hand, our overall algorithm for finding optimal counterfactuals is simple: 


\begin{figure}[h]
  \captionsetup{width=.9\linewidth}
\begin{tcolorbox}[colback = white,arc=1mm, boxrule=0.25mm]
\vspace{3pt}

$\textsc{FindOptimal}_f(x^\star)$:
\begin{itemize}[align=left]
	\item[\textbf{Given:}] Instance $x^\star$ and queries to $f$. 
	\item[\textbf{Output:}] The set $\mathcal{C}$ of all optimal counterfactuls for $x^\star$. 
\end{itemize}
\begin{enumerate}
    \item Initialize $k \gets 1, \mathcal{C} \gets \varnothing$.
    \item While $\mathcal{C} \ne \varnothing$:
    \begin{enumerate}
        \item Let $\mathcal{C} \gets \textsc{FindMinimal}_f(*^d, x^\star, k, \varnothing)$.
        \item $k \gets k + 1$.
    \end{enumerate}
    \item Return $\mathcal{C}$.
\end{enumerate}
\vspace{1pt} 
\end{tcolorbox}
\caption{Algorithm for finding optimal counterfactuals, using {\sc FindMinimal} as its main subroutine.}
\label{fig:FindOptimal}

\end{figure}

\Cref{thm:main-intro} is a straightforward consequence of~\Cref{lem:efficient_IDTs,lem:correctness,lem:runtime}:  

\begin{theorem}[Formal version of~\Cref{thm:main-intro}]
\label{thm:main} 
Let $f : \zo^n\to\zo$ be a monotone function. Then $\textsc{FindOptimal}_f(x^\star)$, using the IDT operations given in $\Cref{lem:efficient_IDTs}$, returns all optimal counterfactuals for $x^\star$ and makes 
\[ S(f)^{O(\Delta_f(x^\star))} \cdot \log d \]
queries to $f$.
\end{theorem} 

\begin{proof}
The correctness of {\sc FindOptimal} follows from that of {\sc FindMinimal}: {\sc FindOptimal} calls $\textsc{FindMinimal}_f(*^d, x^\star, k,\varnothing)$ for $k=1,2,3,\ldots$ until a call returns a nonempty set.  This happens exactly when $k$ reaches $\Delta_f(x^\star)$. 

It remains to analyze the query complexity of {\sc FindOptimal}.  
By \Cref{lem:efficient_IDTs}, the height of $T$ is at most $S(f)^{O(1)}$. Therefore by \Cref{lem:runtime} we have that $\textsc{FindMinimal}_f(*^d,x^\star,k,\emptyset)$ uses at most $S(f)^{O(k)}$  IDT operations and additional queries to $f$.  By~\Cref{lem:efficient_IDTs}, each IDT operation can be supported with $S(f)^{O(1)} \cdot \log d$ queries to $f$.  The overall query complexity of  {\sc FindOptimal} is therefore upper bounded by
\[ \sum_{k=1}^{\Delta_f(x^\star)}  S(f)^{O(k)} \cdot S(f)^{O(1)}\cdot \log d  \le S(f)^{O(\Delta_f(x))} \cdot \log d. \qedhere \] 
\end{proof}

\section{Building efficient implicit decision trees: Proof of~\texorpdfstring{\Cref{lem:efficient_IDTs}}{Lemma 3.1}}
\label{sec:proof-of-lemma} 
\label{sec:build-idt}
In this section, we prove \Cref{lem:efficient_IDTs}. Our implementation of IDTs will be based on known algorithms for finding \textit{certificates}. These certificates will be stored in the auxiliary information of the IDT and will be used to select which features to query in the decision tree. 

\subsection{Certificate complexity and query-efficient certificate finding}
At a high level, a certificate of $f$ on $x^\star\in\zo^d$ is a set of features $W\sse [d]$ that ``witness'' $f$'s value on $x^\star$ in the sense that $f(y) = f(x^\star)$ for all $y\in\zo^d$ that agree with $x^\star$ on the coordinates in $W$.

\begin{definition}[Certificate complexity]
    \label{def:cert-complexity}
    For a model $f : \zo^d \to \zo$ and an input $x^\star$, the {\sl complexity of certifying $f$'s value on $x^\star$}, $C(f,x^\star)$, is the quantity: 
    $$
    \min_{W\sse [d]}\big\{ |W| \colon \text{$f(y)= f(x^\star)$ for all $y$ s.t.~$y_W = x^\star_W$}\big\}.
    $$
    The {\sl certificate complexity of $f$} is the quantity $\ds C(f)\coloneqq \max_{x\in \zo^d} \{ C(f,x)\}$.
\end{definition}

\paragraph{Example.}{ Because the set $W=[d]$ is always a certificate of $f$ on $x^\star$, we typically want \textit{small} certificates satisfying $|W|\le C(f)$. For a specific example, consider the function $f:\zo^5\to\zo$ defined by the decision tree in \Cref{fig:IDT_query_example}. Then $W=\{1,3,5\}$ is a certificate of $f$ on $x^\star=10000$ since $f$ outputs $0$ on any input $y$ satisfying $y_1=1, y_3=0,$ and $y_5=0$. Indeed, for an arbitrary instance $x^\star$ the indices queried along the root-to-leaf path for $x^\star$ constitute a certificate of $f$ on $x^\star$.}

In general, we say $W\sse [d]$ is a certificate of $f$ if there exists some $x\in\zo^d$ such that $W$ is a certificate of $f$ on $x$. Our algorithm relies on a query-efficient procedure {\sc Cert}($f$) which takes a monotone $f$ and returns the features in small certificate of $f$. The features in this certificate will be exactly the features we query in our IDT for $f$. ``Small certificate'' here means having size at most $C(f)^{O(1)}$. For monotone functions, this size bound is equivalently $S(f)^{O(1)}$ using the fact that certificate complexity equals sensitivity for monotone models $f$:

\begin{fact}[\cite{Nis89}]
    \label{fact:sensitivity_equals_cc}
    For a monotone model $f:\zo^d\to\zo$, $S(f)=C(f)$.
\end{fact}

 The authors of \cite{BKLT22} recently developed a query-efficient implementation of \textsc{Cert}$(\cdot)$ for monotone $f$. We restate their result here in terms of $S(f)$ using \Cref{fact:sensitivity_equals_cc}.

\begin{theorem}[Theorem 5 of~\cite{BKLT22}]
    \label{thm:efficient_certification}
    Let $f:\zo^d\to\zo$ be monotone. Then there is an algorithm that computes a certificate of $f$ of size $S(f)^{O(1)}$ with high probability using $S(f)^{O(1)}\cdot\log d$ queries to $f$. 
\end{theorem}

\subsection{Overall structure of \texorpdfstring{$T$}{T}}

\Cref{fig:IDT_query_example} depicts the overall structure of an IDT. Conceptually, the tree can be divided into subtree ``blocks''. For a subtree rooted at $\rho$, this block is the complete binary tree that exhaustively queries all features in $\textsc{Cert}(f_\rho)$, a small certificate for $f_\rho$. One call to $\textsc{Cert}(f_\rho)$ yields a set of features to query in the tree and so we store this set as auxiliary information until we need to call $\textsc{Cert}(\cdot)$ again. The overall depth of the tree is the number of consecutive blocks in a root-to-leaf path times the depth of each block. The depth of each block is the size of the certificate returned by $\textsc{Cert}(\cdot)$ and so \Cref{thm:efficient_certification} ensures that this depth is small. In the next section, we give the implementation details of the IDT operations which yield this IDT structure. Then, we show how to bound the number of consecutive blocks which provides a bound on the overall depth of the IDT.

\begin{figure}[h!]
    \centering
    \begin{tikzpicture}[tips=proper]
        \node[isosceles triangle,
            draw,
            isosceles triangle apex angle=60,
            rotate=90,
            minimum size=3cm] (T1) at (0,0){};
        
        \node[isosceles triangle,
            draw,
            isosceles triangle apex angle=60,
            rotate=90,
            minimum size=3cm,
            anchor=east] (T2) at (T1.220){};
        \node[isosceles triangle,
            draw,
            isosceles triangle apex angle=60,
            rotate=90,
            minimum size=3cm,
            anchor=east] (T3) at (T2.150){};
        \node[isosceles triangle,
            draw,
            isosceles triangle apex angle=60,
            rotate=90,
            minimum size=3cm,
            anchor=east] (T4) at (T3.200){};
        
        \node[isosceles triangle,
            draw,
            dashed,
            isosceles triangle apex angle=60,
            rotate=90,
            minimum size=12.064cm,
            anchor=east] (T5) at (T1.east){};
        
        \draw[black,dotted] (T1.east) .. controls (T1.140) and (T1.280) .. (T1.220) node[fill=white,pos=0.4] (N1) {$\rho_1$};
        \draw[black,dotted] (T2.east) .. controls (T2.357) and (T2.100) .. (T2.150) node[fill=white,pos=0.7] (N2) {$\rho_2$};
        \draw[black,dotted] (T3.east) .. controls (T3.330) and (T3.60) .. (T3.200) node[fill=white,pos=0.7] (N3) {$\rho_3$};
        \draw[black,dotted] (T4.east) .. controls (T4.340) and (T4.30) .. (T4.west) node[fill=white,pos=0.8] (N4) {$\rho_4$};
        
        \draw[color=black] (T4.west) node [below,fill=white] {$x^\star$};
        \node[draw,circle,fill=black,inner sep=1pt] (x) at (T4.west) {};
        
        \draw[color=black] ([xshift=2.0cm,yshift=-0.5cm]T3.west) node [right,fill=white] {\small Restriction corresponding to path segment};
        \draw[stealth-,gray] (N4) edge [out=0, in=180]  ([xshift=2.0cm,yshift=-0.5cm]T3.west);

        \draw[{Stealth[scale=1]}-{Stealth[scale=1]}] ([xshift=3cm]T1.east) to node[midway,fill=white!30,scale=1] {depth $C(f)^{O(1)}$} ([xshift=3cm]T1.west);
        
        
        \draw [gray,decorate,decoration={brace,mirror,raise=3pt,amplitude=4pt}] ([yshift=3cm]T1.left corner) -- (T1.left corner) node [black,pos=0.5,xshift=-0.3cm,left] {\small Certificate of $f$};
        \draw [gray,decorate,decoration={brace,mirror,raise=3pt,amplitude=4pt}] ([yshift=0cm]T1.left corner) -- ([yshift=-3cm]T1.left corner) node [black,fill=white,pos=0.5,xshift=-0.3cm,left] {\small Certificate of $f_{\rho_1}$};
        \draw [gray,decorate,decoration={brace,mirror,raise=3pt,amplitude=4pt}] ([yshift=-3cm]T1.left corner) -- ([yshift=-6cm]T1.left corner) node [black,fill=white,pos=0.5,xshift=-0.3cm,left] {\small Certificate of $f_{\rho_1,\rho_2}$};
        \draw [gray,decorate,decoration={brace,mirror,raise=3pt,amplitude=4pt}] ([yshift=-6cm]T1.left corner) -- ([yshift=-9cm]T1.left corner) node [black,fill=white,pos=0.5,xshift=-0.3cm,left] {\small Certificate of $f_{\rho_1,\rho_2,\rho_3}$};
    \end{tikzpicture}

    \caption{Tracing the root-to-leaf path of an input $x^\star\in\zo^d$ through an implicit decision tree (IDT) separated into blocks. The dashed triangle outline represents the entire tree. The IDT operations allow one to construct the dotted path without constructing the entire tree. Each triangular block corresponds to the subtree exhaustively querying all the features in a certificate of a restriction of $f$. The depth of each block is $C(f)^{O(1)}$, the number of features in a certificate of $f$. The total number of blocks intersecting any root-to-leaf path is at most $C(f)^{O(1)}$ and so the overall depth of the IDT is at most $C(f)^{O(1)}$.}
    \label{fig:path_thru_idt}
\end{figure}
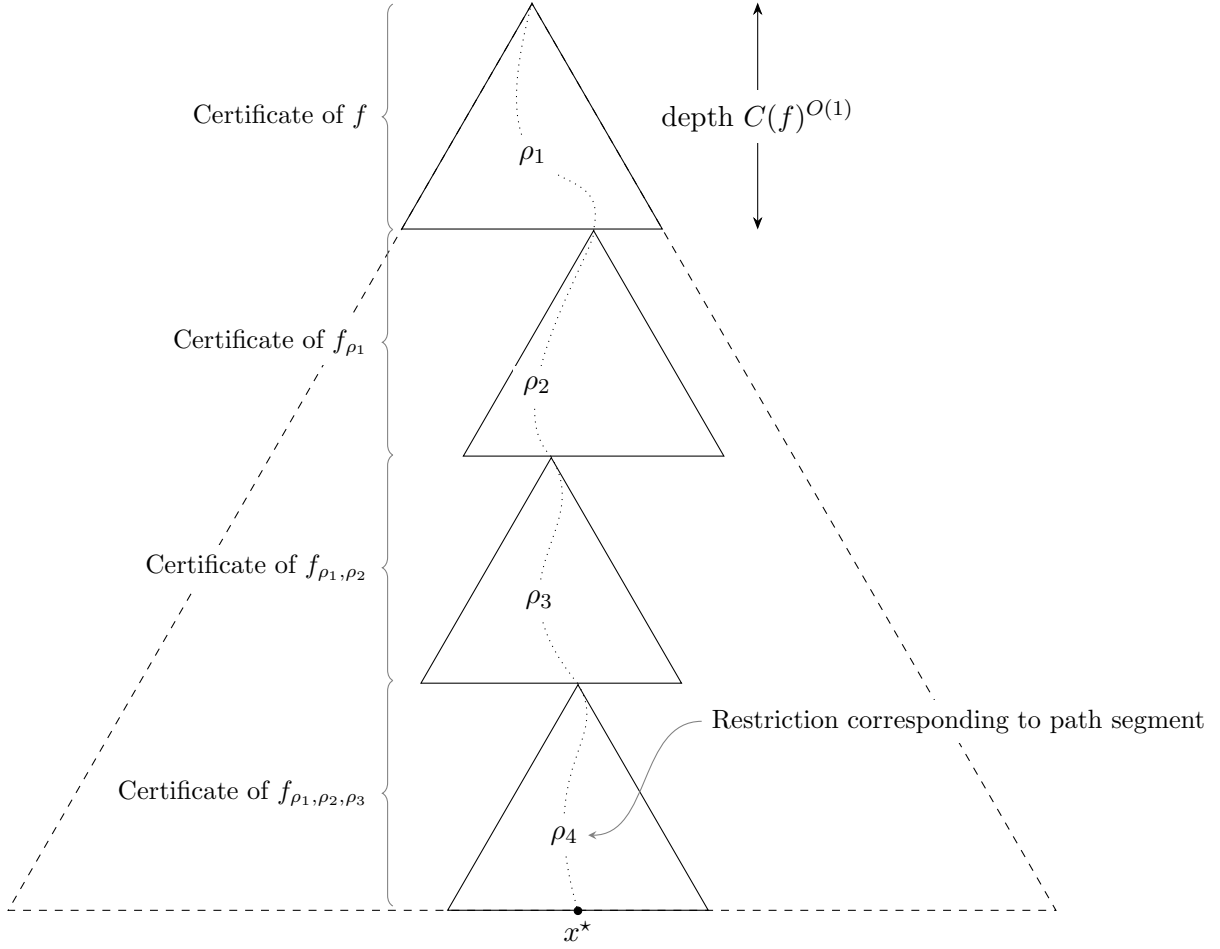

\subsection{Query-efficient implementation of IDT operations for navigating \texorpdfstring{$T$}{T}} 

\paragraph{{\sc Query}$_f(\cdot,\cdot)$.}{We assume there is some ordering on features $[d]$ and hence on elements of a subset $A\subseteq [d]$. Consider the following implementation of {\sc Query}$_f$.

{\sc Query}$_f(\rho, A)$:
\begin{enumerate}
    \item if $A\subseteq \text{Dom}(\rho)$ then let $A'=\textsc{Cert}(f_\rho)$ and return the next feature of $A'$;
    \item otherwise, return the next feature in $A\setminus \text{Dom}(\rho)$ and set $A'=A$. 
\end{enumerate}

This procedure stores a list of features in the auxiliary information, $A$, which is then queried iteratively until all features have been exhausted after which a new set of features is computed using {\sc Cert}. The set $A$ specifies a complete binary tree of depth $|A|$ -- the tree which exhaustively queries all of the features in $A$. In this way, the entire IDT is a collection of complete subtree blocks each of which originates from a certificate output by {\sc Cert}. \Cref{fig:path_thru_idt} illustrates the computation path of a particular input through an IDT using auxiliary information to store subtree blocks.
}

\paragraph{{\sc IsLeaf}$_f(\cdot,\cdot)$.}{
When $f$ is monotone, one can check whether it's the constant function by querying $f$ on the all $0$s input and the all $1$s input. Therefore, to check whether $\rho$ is a leaf, we can query $f_\rho$ on these two inputs and output Yes if and only if $f_\rho$ is determined to be the constant function. 
}

\subsection{Bounding the depth of \texorpdfstring{$T$}{T}} 

We also use the following lemma (implicit in Theorem 6 of~\cite{BKLT22}) which states that $f$ quickly becomes the constant function after being repeatedly restricted by a certificate. Alternatively, consecutive blocks in an IDT quickly terminate in a leaf node. More specifically, suppose $A=\textsc{Cert}(f)$ is a certificate for $f$. So, there is some restriction $\rho:[d]\rightharpoonup \zo$ with $\text{Dom}(\rho)=A$ such that $f_\rho$ is a constant function. Equivalently, there is some way of setting each feature in $A$ to a value in $\zo$ such that $f_\rho$ is the constant function. In general, for an \textit{arbitrary} restriction $\rho:[d]\rightharpoonup \zo$ with $\text{Dom}(\rho)=A$ (an arbitrary setting of each feature in $A$ to $\zo$), $f_\rho$ may be nonconstant. However, one can show $f_\rho$ is ``closer'' to a constant than $f$ in the sense that $f_\rho$ has smaller certificate complexity. The following lemma captures this behavior and states more specifically that $f$ becomes constant after $O(S(f))$ many such restrictions. For more details, see the proof of the lemma in \Cref{appendix:restrict_f_to_constant}. 

\begin{lemma}[Implicit in Theorem 6 of~\cite{BKLT22}]
    \label{lem:restricting_to_constant}
    Let $f:\zo^d\to\zo$ be monotone. Let $\rho_1,\ldots,\rho_\ell$ be a series of restrictions and $A_1,\ldots,A_\ell\subseteq [d]$ be disjoint such that for all $i\in [\ell]$ $\text{Dom}(\rho_i)=A_1\cup\cdots\cup A_i$ and each $A_{i+1}$ is a certificate for $f_{\rho_i}$. Then $\ell=2\cdot S(f)$ suffices to ensure that $f_{\rho_\ell}$ is constant.
\end{lemma}

\subsection{Putting it all together: Proof of \texorpdfstring{\Cref{lem:efficient_IDTs}}{Lemma 3.1}}

We prove that our implementations of {\sc Query}$_f$ and {\sc IsLeaf}$_f$ suffice to prove the lemma statement. First, we show that the implementations require at most $S(f)^{O(1)}\cdot\log d$ queries to $f$, then we show the depth of the IDT is at most $S(f)^{O(1)}$. 

The procedure {\sc IsLeaf} uses $O(1)$ queries to $f$. A call to {\sc Query}$(\rho,A)$ returns either a feature in $A$ or queries {\sc Cert}$(f)$. \Cref{thm:efficient_certification} shows that {\sc Cert}$(f)$ requires at most $S(f)^{O(1)}\cdot\log d$ to $f$ and hence {\sc Query}$(\rho,A)$ requires at most $S(f)^{O(1)}\cdot\log d$ queries to $f$.

It remains to show that the depth of the IDT is at most $S(f)^{O(1)}$. Let $\rho$ be an arbitrary leaf of the IDT. By construction, $\text{Dom}(\rho)$ can be partitioned into blocks of features $A_1,\ldots,A_\ell$ where each $A_i$ is a certificate of $f$ restricted on the features in $A_1\cup\cdots\cup A_{i-1}$. \Cref{lem:restricting_to_constant} ensures that $f$ becomes constant after being restricted by $\ell=S(f)^{O(1)}$ many such certificates. \hfill
\qed

\section{Lower bounds}
\label{sec:lb}

\subsection{Optimality of \texorpdfstring{\Cref{thm:main}}{Theorem 6}}

We prove the following (slightly more general and formal) version of \Cref{thm:lb-intro}.
\begin{theorem}
    \label{thm:lb-monotone}
    Fix any constant $c < 1$ and success probability $\eps > 0$. For any algorithm $\mcA$, and $S,d, \Delta \in \N$ such that $S \leq d$ and $\Delta \leq S^c$, there is a monotone model $f: \zo^d \to \zo$ with $S(f) \leq S$ and instance $x^\star$ with $\Delta_f(x^\star) = \Delta$ such that $\mcA$ must make
    \begin{equation*}
        q = S^{\Omega(\Delta)} + \Omega(\log d)
    \end{equation*}
    queries to find a single optimal counterfactual for $x^\star$ with success probability $\eps$.
\end{theorem}

We'll allow the algorithm $\mathcal{A}$ to use randomness. In order to give lower bounds against randomized algorithms, we'll use the easy direction of Yao's Lemma.

\begin{lemma}[\cite{Yao77}]
    \label{lem:Yao}
    For any $q \in \N$, let $\mathcal{R}_q$ and $\mathcal{D}_q$ be the set of all $q$-query randomized and deterministic algorithms respectively, and let $I$ be the set all of possible pairs $f: \zo^n \to \zo$ and $x^\star \in \zo^n$ (i.e.~instances of the counterfactual problem). 
    
    For any distribution $\mu$ supported on $I$,
    \begin{align*}
        \min_{R \in \mathcal{R}_q} &\max_{(f,x^\star) \in I}[\error_R(f,x^\star)] \\
        \geq &\min_{D \in \mathcal{D}_q} \Ex_{(\boldf,\bx^\star) \sim \mu}[\error_D(\boldf,\bx^\star)]
    \end{align*}
    where $\error_R(f,x^\star)$ is the probability that $R$ does not return an optimal counterfactual for $x^\star$, and $\error_D(f,x^\star) = \Ind[\text{$D$ does not return an optimal counterfactual for $x^\star$}]$.
\end{lemma}

First, we show that $s^{\Omega(\Delta)}$ queries are necessary.
\begin{lemma}
    \label{lem:lb-power}
    For any $\Delta, d, S, q \in \N$ where $S \leq d$, and $\mathcal{A}$ a $q$-query randomized algorithm, there exists some monotone function $f : \zo^d \to \zo$ with $S(f) \leq S$ and input $x^\star \in \zo$ satisfying $\Delta_f(x^\star) = \Delta$ on which $\mathcal{A}$ successfully returns an optimal counterfactual for $x^\star$ with probability at most $\frac{q + 1}{\binom{S}{\Delta}}$.
\end{lemma}
\begin{proof}
    We will apply Yao's Lemma: To apply it, we need to define a distribution over input instances $(\boldf, \bx^\star)$. That distribution is simple. With probability $1$, $\bx^\star = (0, \ldots, 0)$. Then, we sample a uniformly random $\bz$ from the $\binom{S}{\Delta}$ elements of $\zo^S$ with Hamming weight equal to $\Delta$ and set $\boldf:\zo^d \to \zo$ as
    \begin{equation*}
        \boldf(x) \coloneqq \begin{cases}
            1~|~\text{$\ge \Delta + 1$ of the first $s$ coordinates of $x$ are $1$}\\
            1~|~\text{the first $s$ coordinates of $x$ equal $\bz$}\\
            0~|~\text{otherwise}.
        \end{cases}
    \end{equation*}
    
    First, we note that it is always true that $S(\boldf) \leq S$. This is because $\boldf(x)$ only depends on the first $S$-coordinates of $x$, so flipping any of the other $d - S$ coordinates cannot change $\boldf$'s output. Second, we note that the input $\bx'$ being $\bz$ concatenated with $d - S$ many $0$s satisfies $\boldf(x^\star) \neq \boldf(\bx')$ and $|x^\star - \bx'| = \Delta$. Therefore, $\Delta_{\boldf}(x^\star) \leq \Delta$. Furthermore, $\bx'$ is the \emph{only} counterfactual for $x^\star$ at distance $\leq \Delta$, and so the algorithm succeeds if and only if it outputs $\bx'$.
    
    Consider an arbitrary $q$-query \emph{deterministic} algorithm $\mathcal{A}'$.  By Yao's Lemma, it is sufficient to prove that the probability $\mathcal{A}'$ outputs $\bx'$ given the random instance $(\boldf, \bx^\star)$ is at most $\frac{q + 1}{\binom{S}{\Delta}}$. As $\mathcal{A}'$ is deterministic, the queries it makes and answer it outputs are a deterministic function of the answers to its previous queries. Let $g: \zo^d \to \zo$ be defined as
    \begin{equation*}
        g(x) \coloneqq \begin{cases}
            1~|~\text{$\ge \Delta + 1$ of the first $S$ coordinates of $x$ are $1$}\\
            0~|~\text{otherwise}.
        \end{cases}
    \end{equation*}
    We will consider \emph{one} possible execution path of $\mathcal{A}'$, when $\boldf$ and $g$ give the same answer on each query $\mathcal{A}'$ makes. Specifically, the execution path of $\mathcal{A}'$ when every query it makes is answered consistently with $g$. In that case, $\mathcal{A}'$ will make (the deterministic) queries $x^{(1)}, \ldots, x^{(q)}$. We compute
    \begin{align*}
        \Pr&[\text{Any query inconsistent with $g$}] \leq \sum_{j = 1}^q \Pr[\boldf(x^{(j)}) \neq g(x^{(j)})] \tag{Union bound} \\
        &= \sum_{j = 1}^q \Pr[\text{The first $S$ coordinates of $x^{(j)}$ equal $\bz$}] \\
        &\leq \frac{q}{\binom{S}{\Delta}}.
    \end{align*}
    Lastly, if $\mathcal{A}'$ follows the single execution path corresponding to every query it makes being answered consistently with $g$, it will output some answer $x^{(g)}$. This answer is correct if and only if it equals $\bx'$, which occurs with probability at most $\frac{1}{\binom{S}{\Delta}}$ for any choice of $x^{(g)}$. Therefore,
    \begin{align*}
        \Pr[\mathcal{A}'&\text{ returns an optimal counterfactual}] \\
        &=  \Pr[\mathcal{A}'\text{ returns an optimal counterfactual}, \text{All queries consistent with $g$}] \\
        &\quad+\Pr[\mathcal{A}'\text{ returns an optimal counterfactual},\text{Any query inconsistent with $g$}] \\
        &\leq \Pr[x^{(g)} = \bz] +\Pr[\text{Any query inconsistent with $g$}] \\
        &\leq \frac{1}{\binom{S}{\Delta}} + \frac{q}{\binom{S}{\Delta}} = \frac{q+1}{\binom{S}{\Delta}}
    \end{align*}
    The desired result follows from Yao's Lemma.
    
\end{proof}

Next we show that $\log d$ queries are necessary.
\begin{lemma}
    \label{lem:lb-log-d}
    For any $d,q \in \N$ and $\mathcal{A}$ a $q$-query randomized algorithm, there exists a monotone model $f:\zo^d \to \zo$ satisfying $S(f) = 1$ and input $x^\star$ where $\Delta_f(x^\star) = 1$ on which $\mathcal{A}$ successful returns an optimal counterfactual for $x^\star$ with probability at most $\frac{2^q}{d}$.
\end{lemma}
\begin{proof}
    Once again, we will pick a distribution over input instances $(\boldf, \bx^\star)$ and apply Yao's Lemma. With probability $1$, we set $\bx^\star = (0, \ldots, 0)$. Then, we sample a uniformly random $\bi \in [d]$ and set
    \begin{equation*}
        \boldf(x) = x_{\bi}
    \end{equation*}
    First, we note that $S(\boldf) = 1$, as the $f$ is only sensitive on the $\bi^{\text{th}}$ coordinate. Furthermore, the input that is $1$ on the $\bi^{\text{th}}$ coordinate and $0$ everywhere else is a counterfactual for $\bx^\star$ of distance $1$. This is the only optimal counterfactual.
    
    All that remains is to argue that any $q$-query deterministic algorithm fails to return an optimal counterfactual for $\bx^\star$. A $q$-query deterministic algorithm can only output $2^q$ different answers. However, based on the choice of $\bi$, there are $d$ unique choices for the optimal counterfactual. Therefore, the probability the algorithm can output the optimal counterfactual is at most $\frac{2^q}{d}$.
\end{proof}

\begin{proof}[Proof of \Cref{thm:lb-monotone}]
    It's sufficient to show that
    \begin{equation*}
        q = \max(S^{\Omega(\Delta)}, \Omega(\log d))
    \end{equation*}
    where we treat $\eps$ and $c$ as constants. By \Cref{lem:lb-power}, we have that
    \begin{align*}
        q &\geq \eps \cdot \binom{S}{\Delta} - 1 \\
        &\geq \eps\cdot \left(\frac{S}{\Delta}\right)^{\Delta} - 1\\
        &\geq \eps\cdot \left(S^{1-c}\right)^{\Delta} - 1 \tag{$\Delta \leq S^c$} \\
        &= S^{\Omega(\Delta)}.
    \end{align*}
    Then, by \Cref{lem:lb-log-d}, $\frac{2^q}{d} \geq \eps$, and so
    \begin{equation*}
        q \geq \log(\eps d) = \Omega(\log d).
    \end{equation*}
\end{proof}

\subsection{Lower bound against local algorithms}

Recall that an algorithm for finding a counterfactual for $x^\star$ is \emph{local} if its first query is $x^\star$ and every subsequent query differs from a previous query by exactly one feature.

\begin{theorem}
    \label{thm:lb-local-formal}
    For any local search algorithm $\mcA$, $d \in \N$, and $\eps > 0$, there is a monotone model $f: \zo^d \to \zo$ with $S(f) = 1$ and an instance $x^\star$ with $\Delta_f(x^\star)$ such that $\mcA$ must make $q \geq \eps d$ queries in order to find a single optimal counterfactual for $x^\star$.
\end{theorem}
\begin{proof}
    As in the proof of \Cref{lem:lb-log-d}, we apply Yao's Lemma with the distribution with $\bx^{\star} = (0, \ldots, 0)$ with probability $1$ and
    \begin{equation*}
        \boldf(x) = x_{\bi}
    \end{equation*}
    where $\bi \in [d]$ is chosen uniformly at random. Once again, $S(\boldf) = 1$ as $f$ is only sensitive on the $i^{\mathrm{th}}$ and there is a counterfactual from $\bx^\star$ of distance $1$.
    
    Since we are applying Yao's Lemma, it is sufficient to reason about deterministic local algorithms. $\mcA'$ be a $q$-query local and deterministic algorithm. We will consider one execution path of $\mcA'$, namely that in which all its queries return $0$. In that case, it will make the deterministic queries $x^{(1)}, \ldots, x^{(q)}$. As $\mcA'$ is local, the queries $x^{(1)}, \ldots, x^{(q)}$ must all be adjacent, and since $x^{(1)} = (0,\ldots,0)$, the number of coordinates $j \in [d]$ on which $x^{(t)}_j = 1$ for some $t \in [q]$ is at most $q - 1$. The probability that $\boldf(x^{(t)}) = 1$ for some $t \in [q]$ is the probability that $\boldf(x^{(t)})_{\bi}= 1$ for some $t \in [q]$ which is at most $\frac{q-1}{d}$.
    
    Therefore $\mcA'$ follows a single execution path with probability at least $1 - \frac{q-1}{d}$. At the end of that execution path, it outputs a single ``guess" counterfactual $y$. However, $y$ is an optimal counterfactual for $\bx^\star$ only if $y_{\bi} = 1$ and $y_{j} = 0$ for all $j \neq \bi$, which can occur with probability at most $\frac{1}{d}$ for any single $y$. Therefore, the probability $\mcA'$ outputs a correct counterfactual is at most
    \begin{equation*}
        \Pr[y\text{ is an optimal counterfactual for }\bx^\star] + \Pr[\mcA'\text{ does not output }y] \leq \frac{1}{d} + \frac{q-1}{d} = \frac{q}{d}.
    \end{equation*}
    
    By applying Yao's Lemma, we conclude that any randomized $q$-query local algorithm succeeds in finding an optimal counterfactual with probability at most $\frac{q}{d}$. Therefore, to succeed with probability $\eps$, we must have $q \geq \eps d.$
\end{proof}




\section{Functions over general feature spaces} 

 \label{sec:general-feature}
 
 For a monotone model $f:X^d\to\zo$ and an instance $x^\star\in X^d$, we reduce the problem of computing $\mathcal{C}_{f}(x^\star)$ to the problem of computing $\mathcal{C}_{f_{x^\star}}(y)$ where $f_{x^\star}:\zo^d\to\zo$ is a Boolean model defined in terms of $f$ and $x^\star$ and $y\in\zo^d$ is a specific Boolean input to $f_{x^\star}$.

\paragraph{The reduction.}{
    We assume $f:X^d\to\zo$ is monotone with respect to some total ordering on $X$. Write $\overline{X}$ for the top element of $X$ and $\underline{X}$ for the bottom element of $X$. Assuming $f$ is nonconstant, we then have $f(\overline{X}^d)=1$ and $f(\underline{X}^d)=0$. For $x\in X^d$ and $y\in \zo^d$, we define $x\uparrow y, x\downarrow y\in X^d$ as instances satisfying 
    $$
    (x\uparrow y)_i=
    \begin{cases}
        \overline{X} & y_i=1\\
        x_i & y_i=0
    \end{cases}\qquad
    (x\downarrow y)_i=
    \begin{cases}
        x_i & y_i=1\\
        \underline{X} & y_i=0
    \end{cases}
    $$
    for all $i\in [d]$. That is, $x\uparrow y$ is the instance formed by ``snapping'' a subset of features, specified by $y$, to their largest possible value $\overline{X}$. In $x\downarrow y$, features are snapped in the opposite direction. Note that relative to the ordering $\le$ on $X$, we have $x\downarrow y\le x\le x\uparrow y$ for all $y\in \zo^d$. Given a monotone $f:X^d\to\zo$, we write $x\updownarrows_f y$ to denote snapping $x$ in a direction specified by $f(x)$:
    $$
    x\updownarrows_f y=
    \begin{cases}
        x\downarrow y & f(x)=1\\
        x\uparrow y & f(x)=0.
    \end{cases}
    $$
    We will just write $x\updownarrows y$ when $f$ is known from context. 
    We then define the Boolean function $f_x:\zo^d\to\zo$ as
    $$
    f_x(y)=f(x\updownarrows y).
    $$
    Note that $f_x$ is monotone since $y\le y'$ implies $x\updownarrows y\le x\updownarrows y'$. 
}

\paragraph{Properties of $f_x$.}{ 
    We establish two important properties of the reduction $f_x$. The first property states that the set
    $$
    \mathcal{C}_f(x)=\{\Delta(x',x):x'\in X^d\text{ is an optimal counterfactual for }x\text{ with respect to }f\}
    $$
    is equal to the set $\mathcal{C}_{f_x}(f(x)^d)$ where $f(x)^d\in \{0,1\}^d$ is the bit $f(x)$ concatenated with itself $d$ times. The second property states that the sensitivity of $f_x$ is at most the sensitivity of $f$.
    \begin{lemma}
        \label{lem:fx_properties}
        Let $f:X^d\to \zo$ be monotone. Then
        \begin{enumerate}
            \item $\mathcal{C}_f(x)=\mathcal{C}_{f_x}(f(x)^d)$
            \item $S(f_x)\le S(f)$
        \end{enumerate}
        for all $x\in X^d$.
    \end{lemma}

    \begin{proof}
      For (1), fix an input $x\in X^d$. Suppose without loss of generality that $f(x)=1$ (the arguments for $f(x)=0$ are symmetric) so that $f_x(1^d)=f(x\downarrow 1^d)=f(x)=1$. We will show that if $y\in \zo^d$ is a counterfactual for $1^d$ then there exists a counterfactual $x'\in X^d$ for $x$ satisfying $\Delta(x,x')=\Delta(y,1^d)$ and likewise if there is a counterfactual $x'$ for $x$ then there is a counterfactual $y$ for $1^d$ satisfying the same $\Delta(x,x')=\Delta(y,1^d)$. It then follows that $\Delta_f(x)=\Delta_{f_{x}(1^d)}$ and also $\mathcal{C}_f(x)=\mathcal{C}_{f_x}(1^d)$.
Let $y\in \zo^d$ be a counterfactual for $1^d$. Then, $$f_{x}(y)=f(x\downarrow y)=0$$ which shows that $x\downarrow y\in X^d$ is a counterfactual for $x$. Moreover, by definition,
$$
\Delta(x\downarrow y, x)=\Delta(y,1^d).
$$
On the other hand, suppose $x'\in X^d$ is a counterfactual for $x$. Let $y\in \zo^d$ be the instance defined by
$$
y_i=
\begin{cases}
    0 & i\in \Delta(x,x')\\
    1 & i\not\in \Delta(x,x')
\end{cases}
$$
for all $i\in[d]$. Then $x\downarrow y\le x'$ which implies $$f_x(y)=f(x\downarrow y)\le f(x')=0$$ by monotonicity. Hence, $y$ is a counterfactual for $1^d$ and again we have $\Delta(x\downarrow y, x)=\Delta(y,1^d)$. 

For (2), let $y\in\zo^d$ be an instance satisfying $S_{f_x}(y)=S(f_x)$. Then $f_x(y)=f(x\updownarrows y)$. In particular,
    $$
    S(f_x)=S_{f_x}(y)=S_f(x\updownarrows y)\le S(f)
    $$
    as desired. 
    \end{proof}
}

\subsection{Proof of \texorpdfstring{\Cref{thm:main-intro-non-boolean}}{Theorem 2}}
Let $f:X^d\to\zo$ be monotone and $x^\star\in X^d$. Assume without loss of generality $f(x^\star)=1$. Using \Cref{thm:main}, we compute and return the set $\mathcal{C}_{f_{x^\star}}(1^d)$ which queries $f$ at most $S(f_{x^\star})^{O(\Delta_{f_{x^\star}}(1^d))}$ times. \Cref{lem:fx_properties} states that $\mathcal{C}_{f_{x^\star}}(1^d)=\mathcal{C}_{f}(x^\star)$ which ensures we have returned the desired set. Moreover, since these two sets are equal we have that $\Delta_{f_{x^\star}}(1^d)=\Delta_{f}(x^\star)$. We can then write the query bound as 
\begin{align*}
    S(f_{x^\star})^{O(\Delta_{f_{x^\star}}(1^d))}\log d&= S(f_{x^\star})^{O(\Delta_f(x^\star))}\log d\\
    &\le S(f)^{O(\Delta_f(x^\star))}\log d\tag{\Cref{lem:fx_properties}}
\end{align*}
which completes the proof.\hfill{$\square$}

\section{Conclusion} 

We have given a query-optimal algorithm for finding counterfactuals for monotone models. Our algorithm achieves a logarithmic dependence on the dimension $d$, which is exponentially smaller than the previous best known query complexity of $d^{O(\Delta_f(x^\star))}$, achievable by brute-force local search.  

There are several concrete avenues for future work.  The key enabling ingredient in our overall algorithm is our implicit decision tree algorithm for monotone models.  This gives us a way to efficiently navigate an interpretable representation of a black box model $f$ --- a decision tree representation --- even in the case of complex models for which this representation is intractably large.  This is a potentially versatile technique for understanding black box models --- can we use it to efficiently compute other types of explanations?   The query complexity of our algorithm scales with the sensitivity of $f$, a discrete analogue of the Lipschitz constant.  Can our lower bounds be evaded by running our algorithm, or variants of it, for {\sl smoothed} versions $\tilde{f}$ of $f$, obtained from $f$ by adding a small amount of noise?   Finally, while we have focused on sparsity as our notion of distance in this work, it would be interesting to extend our techniques to accommodate other distance functions.

\section*{Acknowledgments}

We thank the ICML reviewers for their useful comments and feedback. 

Guy, Caleb, and Li-Yang are supported by NSF CAREER Award 1942123. Caleb is also supported by an NDSEG fellowship.  Jane is supported by NSF Award CCF-2006664.

\bibliography{bibtex}{}
\bibliographystyle{alpha}

\appendix
\section{Certificate-based restrictions}
\label{appendix:restrict_f_to_constant}

In this section, we prove \Cref{lem:restricting_to_constant}. We say a certificate $S\sse [n]$ of $f$ on $x\in\zo^d$ is a $0$-certificate if $f(x)=0$ and likewise $S$ is a $1$-certificate if $f(x)=1$. We write $C_0(f)=\max_{x\in f^{-1}(0)}\{C(f,x)\}$ to denote the $0$-certificate complexity of $f$ and $C_1(f)=\max_{x\in f^{-1}(1)}\{C(f,x)\}$ for the $1$-certificate complexity of $f$. Our proof will reference the fact that the intersection of a $1$-certificate with a $0$-certificate is necessarily nonempty (since otherwise there would be some instance $x\in\zo^d$ having both a $0$-certificate and a $1$-certificate). 

\begin{fact}
\label{fact:cert_intersections}
Let $S_0$ be a $0$-certificate of $f:\zo^d\to\zo$ and let $S_1$ be a $1$-certificate of $f$. Then $S_0\cap S_1\neq\varnothing$. 
\end{fact}

\begin{proof}[Proof of \Cref{lem:restricting_to_constant}]
  Let $\rho$ be a restriction of $f:\zo^d\to\zo$ such that $\text{Dom}(\rho)=\textsc{Cert}(f)$. We'll show that $C_0(f_\rho)+C_1(f_\rho)\le C_0(f)+C_1(f)-1$. The result then follows by induction. Suppose without loss of generality that $\text{Dom}(\rho)$ is a $1$-certificate of $f$. Then we claim $C_0(f_\rho)\le C_0(f)-1$. Consider any $x\in f_\rho^{-1}(0)$. Consider the string $x'\in\zo^d$ formed by inserting $\rho$ into the string $x$ so that $f(x')=f_\rho(x)$. Let $S_0$ be a $0$-certificate of $f$ on $x'$ with $|S_0|\le C_0(f)$. Then $S_0\setminus\text{Dom}(\rho)$ is a $0$-certificate of $f_\rho$ on $x$. Moreover, $S_0\cap \text{Dom}(\rho)\neq\varnothing$ by \Cref{fact:cert_intersections} and so $|S_0\setminus \text{Dom}(\rho)|\le |S_0|-1\le C_0(f)-1$. Since $x$ is any arbitrary $0$-input to $f_\rho$, we have that $C_0(f_\rho)\le C_0(f)-1$ as desired. 

  It follows each such restriction $\rho$ decreases $C_0(f)+C_1(f)$ by at least $1$. Thus, after at most $2C(f)\ge C_0(f)+C_1(f)$ restrictions $f$ becomes the constant function. 
\end{proof}

\end{document}